\documentclass[oneside]{article}
\usepackage{caption}
\usepackage{qic,epsfig}
\usepackage{amsfonts}
\usepackage[cmex10]{amsmath}
\usepackage{tabularx}
\usepackage{url}
\renewcommand{\thefootnote}{\fnsymbol{footnote}}  

\begin{document}
\setlength{\textheight}{8.0truein}    

\runninghead{CONSTRUCTIONS OF $q$-ARY ENTANGLEMENT-ASSISTED QUANTUM $\ldots$}
            {Jihao Fan, Hanwu Chen, Juan Xu}

\normalsize\textlineskip
\thispagestyle{empty}
\setcounter{page}{423}

\copyrightheading{16}{5\&6}{2016}{0423--0434}

\vspace*{0.88truein}

\alphfootnote

\fpage{423}

\centerline{\bf
CONSTRUCTIONS OF $q$-ARY ENTANGLEMENT-ASSISTED QUANTUM}
\vspace*{0.035truein}
\centerline{\bf MDS CODES WITH MINIMUM DISTANCE GREATER THAN $q+1$
}
\vspace*{0.37truein}
\centerline{\footnotesize
JIHAO FAN}
\vspace*{0.015truein}
\centerline{\footnotesize\it Department of Computer Science and Engineering, Southeast University}
\baselineskip=10pt
\centerline{\footnotesize\it Nanjing, Jiangsu 211189, China
}
\centerline{\footnotesize\it \url{fanjh12@seu.edu.cn}}
\vspace*{10pt}
\centerline{\footnotesize
HANWU CHEN}
\vspace*{0.015truein}
\centerline{\footnotesize\it Department of Computer Science and Engineering, Southeast University}
\baselineskip=10pt
\centerline{\footnotesize\it Nanjing, Jiangsu 211189, China}
\centerline{\footnotesize\it \url{hw_chen@seu.edu.cn}}
\vspace*{10pt}
\centerline{\footnotesize
JUAN XU}
\vspace*{0.015truein}
\centerline{\footnotesize\it College of Computer Science and Technology, Nanjing University of Aeronautics and Astronautics}
\baselineskip=10pt
\centerline{\footnotesize\it Nanjing, Jiangsu 210016, China}
\centerline{\footnotesize\it \url{juanxu@nuaa.edu.cn}}
\vspace*{0.225truein}
\publisher{(received date)}{(revised date)}

\vspace*{0.21truein}

\abstracts{
The entanglement-assisted stabilizer formalism provides a useful framework for con-
structing quantum error-correcting codes (QECC), which can transform arbitrary classical linear
codes into entanglement-assisted quantum error correcting codes (EAQECCs) by using
  pre-shared entanglement between the sender and the receiver.
In this paper, we
construct five classes of entanglement-assisted quantum MDS (EAQMDS) codes  based on classical MDS
codes by exploiting one or more pre-shared maximally entangled states. We show that these EAQMDS codes have much larger minimum distance than the standard  quantum MDS (QMDS)  codes of the same length, and three classes of these EAQMDS codes    consume only one pair of maximally entangled states.
}{}{}

\vspace*{10pt}

\keywords{Entanglement-assisted quantum error-correcting codes, Quantum error-correcting codes,
Maximal-distance-separable (MDS) codes, Maximally entangled state}
\vspace*{3pt}
\communicate{to be filled by the Editorial}

\vspace*{1pt}\textlineskip    
\section{Introduction}
Quantum error-correcting codes (QECC) play a key role  in protecting quantum information from decoherence and quantum noise.
The theory of quantum stabilizer codes allows one to import  classical additive codes that satisfy certain dual-containing relationship for use as a QECC  \cite{calderbank1998quantum,gottesman1997stabilizer, ketkar2006nonbinary}.
Recently, a more general framework called   entanglement-assisted stabilizer formalism  was developed to construct QECCs with the help of pre-shared entanglement between the sender and the receiver \cite{brun2006correcting}.    This framework has the advantage that it allows to construct QECCs from arbitrary classical linear codes, without the dual-containing constraint.
Currently,  many  works have  focused on  the  construction  of binary EAQECCs based on classical binary or quaternary linear codes, see \cite{hsieh2009entanglement,hsieh2011high,fujiwara2010entanglement, fujiwara2013characterization,wilde2014entanglement,lu2014entanglement},  since binary QECCs
might be the most useful ones in the future quantum computers and quantum communications. However,   nonbinary cases have received less attention. Nonbinary EAQECCs would be useful in some quantum communication protocols \cite{lidar2013quantum,wilde2013quantum}.  Just as in  the classical error-correcting codes (ECC) and standard QECCs,   EAQECCs
over higher alphabets  can be used for constructing easily decodable binary
 EAQECCs by using concatenation technology \cite{macwilliams1977theory,grassl1999quantum}. Furthermore, nonbinary QECCs and EAQECCs,  especially nonbinary  quantum MDS (QMDS) codes and  entanglement-assisted quantum MDS (EAQMDS) codes,  are of  significantly theoretical interest, since QMDS codes and EAQMDS codes can achieve the quantum
Singleton bound \cite{ketkar2006nonbinary} and the entanglement-assisted quantum
Singleton bound \cite{brun2006correcting}, respectively.

Let $q$ be  a prime power.  We use $\mathcal{Q}=[[n,k,d]]_q$ to denote a standard $q$-ary QECC   of length $n$ with size $q^k$ and minimum distance $d$.  Then $\mathcal{Q}$ is a $q^k$-dimensional subspace of the $q^n$-dimensional Hilbert space $(\mathbb{C}^q)^{\otimes n}$, which can detect up to $d-1$ and correct up to $\lfloor(d-1)/2\rfloor$ quantum errors.
The parameters of $\mathcal{Q}$ have to satisfy the
quantum Singleton bound: $k \leq n - 2d +2$ in \cite{ketkar2006nonbinary}.
If $\mathcal{Q}$ attains the quantum Singleton bound, then it is called a quantum maximum-distance-separable
(MDS) code. According to the MDS conjecture in \cite{ketkar2006nonbinary}, the maximal length of a QMDS code cannot exceed $q^2+1$, i.e., $n\leq q^2+1$,
except for  the trivial and some special cases in \cite{li2008quantum},  and except for the  existence of QMDS codes with  parameters $[[q^2+2,q^2-4,4]]_q$ for $q=2^m$ shown in \cite{grassl2015quantum}.
As mentioned in \cite{jin2013construction},  QMDS codes of length up to $q + 1$ have been constructed for all possible
dimensions, see \cite{rotteler2004quantum, grassl2004optimal}.  
However, the problem of constructing QMDS codes with length $n$ greater than $q+1$   is
much more difficult.
 Many QMDS codes with certain lengths between $q+1$ and $q^2+1$ have been obtained,
see \cite{li2010construction,jin2010application,la2011new, kai2013new,kai2013constacyclic,zhang2014new,wang2015new,chen2015application}.
  Up to now, the minimum distance of all known nontrivial $q$-ary QMDS codes is less than
or equal to $  q+1$, except for  a few sporadic QMDS codes with large minimum distance in \cite{grassl2015quantum}.  It seems very difficult to improve this limit by using the standard Euclidean or Hermitian construction.

Inspired by these works, in this paper, we propose several constructions of  EAQMDS codes based on classical MDS codes, and we get new  $q$-ary EAQMDS codes with minimum distance greater than $q+1$ for some certain code lengths, while consuming  a few pre-shared maximally entangled states.
If we denote a $q$-ary EAQECC  by $[[n,k,d;c]]_q$, where $c$ is the
number of maximally entangled states required, we get five classes of EAQMDS codes with parameters:

\begin{romanlist}
\item $[[q^2+1,q^2-2d+4,d;1]]_q$,
where $q$ is a prime power, $2\leq d\leq2q$ is an even integer.
\item $[[q^2,q^2-2d+3,d;1]]_q$,
where $q$ is a prime power, $q+1\leq d\leq2q-1$.
\item $[[q^2-1,q^2-2d+2,d;1]]_q$,
where $q$ is a prime power, $2\leq d\leq2q-2$.
\item
$[[\frac{q^2-1}{2},\frac{q^2-1}{2}-2d+4,d;2]]_q$, where  $q$ is an odd prime power,
$\frac{q+1}{2}+2 \leq d\leq  \frac{3}{2}q-\frac{1}{2}$.
\item
$[[\frac{q^2-1}{t},\frac{q^2-1}{t}-2d+t+2,d;t]]_q$, where $q$ is an odd prime power with $t|(q+1)$, $t\geq3 $ is an odd integer, and
$\frac{(t-1)(q+1)}{t}+2 \leq d\leq  \frac{(t+1)(q+1)}{ t}-2$.
\end{romanlist}

EAQMDS codes in (i)-(v) have minimum distance upper limit greater than $q+1$ by consuming a few pre-shared maximally entangled  states.  In particular,
each code in (i)-(iii) has nearly double  minimum distance upper limit of the standard  QMDS code  of the same length  constructed so far,  and   consumes only one pair of maximally  entangled states.  This means that these codes have much better error-correction abilities than the standard QMDS codes of the same length and consume little entanglement. 

This paper is organized as follows. In Section 2, we introduce some basic notations and
definitions of classical ECCs and EAQECCs. We propose  several constructions of EAQMDS codes in Section 3. The conclusion is given in Section 4.

\section{Preliminaries}
\noindent

Firstly, we review some basic results of classical RS codes, constacyclic  codes and several formulas for EAQECCs. For details
on classical ECCs and EAQECCs, see the literature \cite{macwilliams1977theory,berlekamp1968algebraic} and \cite{brun2006correcting,lidar2013quantum,wilde2008optimal}, respectively.

Let $p$ be a prime number and $q$ a power of $p$, i.e., $q=p^r$ for some $r>0$.
$\mathbb{F}_{q^2}$ denotes the finite field with $q^2$ elements.
For any $a\in\mathbb{F}_{q^2}$, we denote by $\overline{a}=a^q$   the conjugation of $a$.
For two vectors $\mathbf{x}=(x_1,x_2,\ldots,x_n)$ and $\mathbf{y}=(y_1,y_2,\ldots,y_n)\in\mathbb{F}_{q^2}^n$, their Hermitian inner product is defined as
\[
\langle\mathbf{x},\mathbf{y}\rangle_h=\sum_{i=1}^{n}\overline{x_i}y_i=\overline{x_1}y_1+\overline{x_2}y_2+\cdots+\overline{x_n}y_n.
\]
Let $\mathcal{C}=[n,k]$ be a $q^2$-ary linear code  of length $n$ and dimension $k$. The Hermitian dual code of $\mathcal{C}$ is defined as
\[
\mathcal{C}^{\bot_h}=\{\mathbf{x}\in\mathbb{F}_{q^2}^n|\langle\mathbf{x},
\mathbf{y}\rangle_h=0,\forall\mathbf{y}\in\mathcal{C}\}.
\]
If $\mathcal{C}\subseteq\mathcal{C}^{\bot_h}$, then $\mathcal{C}$ is called a Hermitian  self-orthogonal code. On the contrary, if
 $\mathcal{C}^{\bot_h}\subseteq\mathcal{C}$, then $\mathcal{C}$ is called a Hermitian  dual-containing code. Let $H=\left(a_{ij}\right)_{(n-k)\times n}$ be the parity check matrix of $\mathcal{C}$ over $\mathbb{F}_{q^2}$ with indices $1\leq i\leq n-k$ and $1\leq j\leq n$,  then the Hermitian conjugate of $H$ is defined as
\[
H^\dag=\left(\overline{a_{ji}}\right)_{n\times (n-k)},
\]
{ where  the dagger ($\dag$) denotes the conjugate transpose operation over matrices in $\mathbb{F}_{q^2}$.

 A Reed-Solomon code (denoted by $\mathcal{RS}(n,r)$) over $\mathbb{F}_{q^m}$
is a cyclic code of length $n=q^m-1$ with roots $\alpha$, $\alpha^{2},\ldots,\alpha^{r-1}$,
where $r$ is an integer with $1\leq r\leq n-2$, $\alpha$ is a
primitive element of $\mathbb{F}_{q^m}$. Its generator polynomial is $g(x)=(x-\alpha)(x-\alpha^2)\cdots(x-\alpha^{r-1})$. The parameters
of $\mathcal{RS}(n,r)$ are $[n,k,d]_{q^m}$, where $k=n-r+1$, $d=r$. The parity check matrix of $\mathcal{RS}(n,r)$ is given by
\begin{eqnarray}
\label{parity check matrix of general RS codes}
H_{\mathcal{RS}(n,r)}=
\left(
\begin{array}{cccc}
1&\alpha &\cdots&\alpha^{ n-1 }\\
1&\alpha^{2}&\cdots&\alpha^{2(n-1) }\\
\vdots&\vdots&\vdots&\vdots\\
1&\alpha^{r-1}&\cdots&\alpha^{(r-1)(n-1)}
\end{array}
\right).
\end{eqnarray}

Let $\lambda$  be a nonzero element of $\mathbb{F}_{q^2}$, then a linear code $\mathcal{C}$ of length $n$ over
$\mathbb{F}_{q^2}$ is said to be $\lambda$-constacyclic if $(\lambda c_n,c_1,\ldots,c_{n-1})\in \mathcal{C}$
for every $(c_1, c_2,\ldots,c_{n})\in \mathcal{C}$. If $\lambda=1$, $\mathcal{C}$ is a cyclic code.
If $\lambda=-1$, $\mathcal{C}$ is called a negacyclic code.
We assume that $\gcd(n,q^2)=1$. A codeword $(c_1, c_2,\ldots,c_{n})\in \mathcal{C}$ is identified with its
polynomial representation $c(x)=c_0+c_1x+\cdots+c_{n-1}x^{n-1}$. It is easy to find that
a $\lambda$-constacyclic code $\mathcal{C}$ of length $n$ over $\mathbb{F}_{q^2}$ is an ideal of the quotient ring
$\mathbb{F}_{q^2}[x]/\langle x^n-\lambda\rangle$. It is known that $\mathcal{C}$ is generated by a monic divisor $g(x)$
of $x^n-\lambda$. The polynomial $g(x)$ is called the generator polynomial of the
code $\mathcal{C}$. Let $\lambda\in\mathbb{F}_{q^2}$ be a primitive $r$th root of unity.
Let $\eta$ denote a primitive $rn$th
root of unity (exists in some extension field) such that $\eta^n=\lambda$. Let $\zeta=\eta^r$ be a
primitive $n$th root of unity. It follows from \cite{krishna1990pseudocyclic} that
the roots of $x^n-\lambda$ are $\{\eta\zeta^i=\eta^{1+ri}|0\leq i\leq n-1\}$.
Denote   $\Omega=\{1+ri|0\leq i\leq n-1\}$. Then the
defining set of a $\lambda$-constacyclic code $\mathcal{C}$ with generator polynomial $g(x)$ is
$Z=\{i\in\Omega|g(\eta^i)=0\}$. It is easy to see that the defining set $Z$ is a
union of some $q^2$-cyclotomic cosets modulo $rn$. There exist the following BCH bound
for cyclic codes and the generalized BCH bound for $\lambda$-constacyclic codes.

\newtheorem{example}{Example}

\begin{lemma}[\it{\cite{macwilliams1977theory}, Ch.7}]
\label{BCHBound_cyclic}
Let $\mathcal{C}$ be a cyclic code of length $n$ over $\mathbb{F}_{q^2}$. Let $\alpha\in\mathbb{F}_{q^2}$ be a primitive $n$-th root of unity. Suppose that $\mathcal{C}$ has
generator polynomial $g(x)$ such that  for some integers $b\geq0$ and $\delta\geq1$,
$g(\alpha^b)=g(\alpha^{b+1})=\cdots=g(\alpha^{b+\delta-2})=0$, that is, the code has a string of
$\delta-1$ consecutive powers of $\alpha$ as zeros. Then the minimum distance
of $\mathcal{C}$ is at least $\delta$.
\end{lemma}

\begin{lemma}[\it{\cite{krishna1990pseudocyclic}, Lemma 4}]
\label{BCHBound_constacyclic}
Let $\mathcal{C}$ be a
$\lambda$-constacyclic code of length $n$ over $\mathbb{F}_{q^2}$, where $\lambda\in\mathbb{F}_{q^2}$ is a primitive $r$th root of unity. Suppose that the generator polynomial $g(x)$
of $\mathcal{C}$ has the elements $\{\eta^{1+ri}|i_0\leq i\leq i_0+d-2\}$ as roots,
where $\eta$ is a primitive $rn$th root of unity, $i_0$ is an integer.
Then the minimum distance of $\mathcal{C}$ is at least $d$.
\end{lemma}

The following lemma gives a sufficient and necessary
condition for a $q^2$-ary $\lambda$-constacyclic code to be Hermitian dual-containing.

\begin{lemma}[\it{\cite{kai2013constacyclic}, Lemma 2.2}]
\label{Hermitian_dual_coset}
Let $\mathcal{C}$ be a  $\lambda$-constacyclic code of length $n$
over $\mathbb{F}_{q^2}$ with defining set $Z$ and
let $\lambda\in\mathbb{F}_{q^2}$ be a primitive $r$th root of unity. Then $\mathcal{C}$ is a
Hermitian dual-containing code if and only if $Z\cap Z^{-q}=\emptyset$ where
$Z^{-q} = \{-qz \pmod{rn} | z\in Z\}$.
\end{lemma}

An $[[n,k,d;c]]_q$ EAQECC encodes $k$ information qudits
into $n$ channel qudits with the help of $c$ pairs of maximally
 entangled states.  The minimum distance is $d$.
One
of the focuses of the construction of EAQECCs is to determine  the number of maximally entangled pairs required for the encoding.
For example, the optimal number of  entangled pairs required by an arbitrary binary EAQECC is given in \cite{wilde2008optimal}.

\begin{theorem}[\it{\cite{wilde2008optimal}, Theorem 1}]
\label{ARBEAQECC}
Suppose that an EAQECC is constructed from generators corresponding to the rows in a
quantum check matrix
\[
H=[H_Z|H_X],
\]
where $H$ is an $[(n-k)\times2n]$-dimensional binary matrix representing
the quantum code (see \cite{gottesman1997stabilizer,nielsen2000quantum}), and both $H_Z$ and $H_X$ are
$[(n-k)\times n]$-dimensional binary matrices. Then the resulting
code is an $[n,k+c;c]$ entanglement-assisted code and requires
$c$ ebits, where
\begin{equation}
\label{binaryebits}
c = rank(H_XH_Z^T+H_ZH_X^T)/2
\end{equation}
and addition is binary.
\end{theorem}

Several formulas for different EAQECCs are given as corollaries in  \cite{wilde2008optimal}.
Similar results are also
available for nonbinary EAQECCs.
According to \cite{wilde2008optimal},
a formula similar to (\ref{binaryebits}) holds for $q$-ary EAQECCs by using $q$-dimensional entangled pairs. The number of the corresponding entangled pairs is given by
\begin{equation}
\label{edits}
c = rank(H_XH_Z^T-H_ZH_X^T)/2
\end{equation}
and subtraction is in the finite field $\mathbb{F}_q$. There are the following corollaries for   general EAQECCs.

\begin{corollary}[\it{\cite{wilde2008optimal}}]
\label{corollary_EAQECC}
Let $H$ be the parity check matrix of an $[n,k,d]_{q^2}$
 classical linear  code  over $\mathbb{F}_{q^2}$. Then an $[[n,2k-n+c,d;c]]_q$ EAQECC can be obtained, where
$c=rank(HH^\dagger)$ is the number of maximally entangled states required.
\end{corollary}

\begin{corollary}[\it{EA-Singleton Bound, \cite{brun2006correcting}}]
An $[[n,k,d;c]]_q$ EAQECC satisfies
\begin{equation}
\label{EASingleton}
n+c-k\geq2(d-1),
\end{equation}
where $0\leq c\leq n-1$.
\end{corollary}

\setcounter{footnote}{0}
\renewcommand{\thefootnote}{\alph{footnote}}
\section{Constructions of $q$-ary EAQMDS codes}
\noindent

 A classical linear MDS code can lead to an EAQECC that meets the
  corresponding EA-Singleton bound \cite{brun2006correcting}. The main task is to
  determine the number of maximally entangled pairs that   required.
For the $q$-ary QMDS code of length $n$, the construction problem has been
completely solved when  length   $n\leq q+1$, see  \cite{rotteler2004quantum, grassl2004optimal}. Therefore, we do not need to consume extra entanglement resources for the construction when  length $n\leq q+1$.
 However, the introduction of a certain amount of pre-shared   entanglement  is useful for the case when length $n>q+1$, since  we may have more variety for the parameters
  of   EAQMDS codes than those of   standard QMDS codes.

\subsection{EAQMDS codes  based on cyclic MDS codes}
\noindent

We take $\mathcal{C}$ as a $q^2$-ary cyclic code over $\mathbb{F}_{q^2}$ of length $n$, where $n|q^2+1$.
Then the $q^2$-cyclotomic coset modulo $n$ containing $i$ is denoted by
$C_i=\{i,iq^2,iq^4,\ldots,iq^{2(m_i-1)}\}$, where $m_i$ is the smallest positive
integer such that $q^{m_i}i=i  \pmod{  n}$. The following result gives the $q^2$-cyclotomic cosets modulo $n$.

 \begin{lemma}[\it{\cite{la2011new}, Lemma 4.1}]
\label{cyclotomic1}
Let $n|q^2+1$  and let $s=\lfloor\frac{n}{2}\rfloor$. If $n$ is odd,
then the $q^2$-cyclotomic cosets modulo $n$
containing integers from 0 to $n$ are: $C_0=\{0\}$,  $C_{i}=\{i,-i\}=\{i,n-i\}$, where $1\leq i\leq s$.
If $n$ is even,
then the $q^2$-cyclotomic cosets modulo $n$
containing integers from 0 to $n$ are: $C_0=\{0\}$, $C_{s}=\{s\}$ and $C_{i}=\{i,-i\}=\{i,n-i\}$, where $1\leq i\leq s-1$.
\end{lemma}

\begin{lemma}
\label{Rank1}
Let $n|q^2+1$ and $s=\lfloor\frac{n}{2}\rfloor$. 
Let $\mathcal{C}$ be a $q^2$-ary cyclic code of length $n$
with defining set $Z=\cup_{i=0}^{\delta}C_{i}$, where
$1\leq\delta\leq\delta_{\max}=\lfloor\frac{n}{q+1}\rfloor$, and let  $H$ be
the parity check matrix of $\mathcal{C}$ over $\mathbb{F}_{q^2}$,
then $rank(HH^\dag)=1$.
\end{lemma}
\begin{proof}
We divide the defining set $Z$ of $\mathcal{C}$ into two mutually disjoint subsets, i.e.,  $Z=C_0\cup Z_1$, where $Z_1=\cup_{i=1}^{\delta}C_{i}$. Let $\mathcal{C}_1$ be  a $q^2$-ary cyclic
code of length $n$ with defining set $Z_1$. 
We show $\mathcal{C}_1^{\perp_h}\subseteq \mathcal{C}_1$. 
Suppose that  $\mathcal{C}_1$ is not a Hermitian dual-containing code,
then $Z_1\cap Z_1^{-q}\neq\emptyset$ by Lemma \ref{Hermitian_dual_coset}. There exist $i$ and $j$,
where $1\leq i,j\leq\delta_{\max}$, such that $i=-qj\pmod{n}$ or $i=qj\pmod{n}$.
If the first case holds, it follows that $q+1\leq i+qj<n$, which is a contradiction. If the second
case holds, it follows that $1\leq i\leq\delta_{\max}< q\leq qj\leq q\delta_{\max}<n$, which is
also a contradiction. Therefore, we have $\mathcal{C}_1^{\perp_h}\subseteq \mathcal{C}_1$.
Let the parity check matrix of $\mathcal{C}_1$
over $\mathbb{F}_{q^2}$ be $H_1$, then $H_1H_1^\dag=0$. It is easy to see that
the parity check matrix of $\mathcal{C}$ over $\mathbb{F}_{q^2}$ is given by
$H=\left(\begin{array}{c}h_0\\H_1\end{array}\right)$, where $h_0=(1,1,\ldots,1)$. Since $n|q^2+1$,
then we have $h_0h_0^\dag\neq0$. It is obvious that $C_0\cap Z_1^{-q}=\emptyset$,
and it follows that $h_0H_1^\dag=0$. Therefore,  the rank of
$ HH^\dag $ is equal to $1$.
\qed
\end{proof}

\begin{theorem}
Let $n|q^2+1$. There exists an EAQMDS code  with parameters
\[
[[n,n-2d+3,d;1]]_q,
\]
where  $2\leq d\leq2\lfloor\frac{n}{q+1}\rfloor+2$ is an even integer. 
\end{theorem}
\begin{proof}
Let $\mathcal{C}$ be a cyclic code of length $n$ with defining set $Z=\cup_{i=0}^{\delta}C_{i}$,
where $0\leq\delta\leq\delta_{\max}=\lfloor\frac{n}{q+1}\rfloor$. From Lemma \ref{cyclotomic1}, we know
that the defining set $Z$ consists of $2\delta+1$
consecutive integers $\{-\delta,-\delta+1,\ldots,-1,0,1,\ldots,\delta-1,\delta\}$.
Then the dimension of $\mathcal{C}$ is dim $\mathcal{C}=n-2\delta-1$. From the BCH bound for cyclic codes in Lemma \ref{BCHBound_cyclic}, we know that the minimum distance
of $\mathcal{C}$ is at least $2\delta+2$. Then $\mathcal{C}$ has parameters  $[n,n-2\delta-1,\geq2\delta+2]_{q^2}$.
Combining Corollary \ref{corollary_EAQECC}, Lemma \ref{Rank1} and the EA-Singleton
bound, we can obtain  an   EAQMDS code with parameters
$[[n,n-4\delta-1,2\delta+2;1]]_q$. Let $d=2\delta+2$, then we have $2\leq d\leq2\delta_{\max}+2=
2\lfloor\frac{n}{q+1}\rfloor+2$. \qed
\end{proof}

Let $n=q^2+1$, then we can get the following EAQMDS code  with minimum distance greater than $q+1$.

\begin{corollary}
\label{EAQMDS11}
There exists an EAQMDS code with parameters
\[
[[q^2+1,q^2-2d+4,d;1]]_q,
\]
where $q$ is a prime power, $2\leq d\leq 2q$ is an even integer.
\end{corollary}

\begin{example}
Let $q=4$, then $n=q^2+1=17$. Applying Corollary \ref{EAQMDS11}, we get two EAQMDS codes with minimum
distance greater than $q+1=5$ whose parameters are $[[17,8,6;1]]_4$, $[[17,4,8;1]]_4$.
\end{example}

If we consider  cyclic codes whose lengths satisfy  $n|q^2-1$,
then the corresponding $q^2$-cyclotomic coset modulo $n$ containing $i$ is
$C_i=\{i\}$, $0\leq i\leq n-1$.

\begin{lemma}
\label{Rank2}
Let  $n|q^2-1$.  
Let $\mathcal{C}$ be a $q^2$-ary cyclic code of length $n$
with defining set $Z=\cup_{i=-\delta}^{\delta}C_{i}$, where
$1\leq\delta\leq\delta_{\max}=\lfloor\frac{n}{q+1}\rfloor-1$, and let  $H$ be
the parity check matrix of $\mathcal{C}$ over $\mathbb{F}_{q^2}$,
then $rank(HH^\dag)=1$.
\end{lemma}
\begin{proof}
We divide the defining set $Z$ of $\mathcal{C}$ into three mutually disjoint subsets, i.e.,  $Z=Z_1\cup C_0\cup Z_2$, where $Z_1=\cup_{i=-\delta}^{-1}C_{i}$ and $Z_2=\cup_{i=1}^{\delta}C_{i}$. Let $\mathcal{C}_1$ and $\mathcal{C}_2$ be two $q^2$-ary cyclic
codes of length $n$ with defining sets $Z_1$ and $Z_2$, respectively. It is easy to verify  that
there are
$\mathcal{C}_1^{\perp_h}\subseteq \mathcal{C}_1$, $\mathcal{C}_2^{\perp_h}\subseteq \mathcal{C}_2$
and $\mathcal{C}_1^{\perp_h}\subseteq \mathcal{C}_2$. 
Let the parity check matrices  of $\mathcal{C}_1$ and $\mathcal{C}_2$ over $\mathbb{F}_{q^2}$ be $H_1$ and $H_2$, respectively,
then we have $H_1H_1^\dag=0$, $H_2H_2^\dag=0$ and $H_1H_2^\dag=0$. Then the parity check matrix of $\mathcal{C}$ over $\mathbb{F}_{q^2}$ is given by
$H=\left(\begin{array}{c}H_1\\h_0\\H_2\end{array}\right)$, where $h_0=(1,1,\ldots,1)$. Since $n|q^2-1$,
then  we have $h_0h_0^\dag\neq0$. It is  obvious that $C_0\cap Z_1^{-q}=\emptyset$ and $C_0\cap Z_2^{-q}=\emptyset$, and it follows
that $h_0H_1^\dag=0$ and $h_0H_2^\dag=0$. Therefore, the rank of
$HH^\dag$ is equal to 1.\qed
\end{proof}

\begin{theorem}
\label{n=q^2-1}
Let $n|q^2-1$. There exists an EAQMDS code  with parameters
\[
[[n,n-2d+3,d;1]]_q,
\]
where   $2\leq d\leq2\lfloor\frac{n}{q+1}\rfloor$. 
\end{theorem}
\begin{proof}
Let $\mathcal{C}$ be a cyclic code of length $n$ with defining set $Z=\cup_{i=-\delta}^{\delta}C_{i}$,
where $0\leq\delta\leq\delta_{\max}=\lfloor\frac{n}{q+1}\rfloor-1$.
Then the defining set $Z$ which consists of $2\delta+1$
consecutive integers  is given by $\{-\delta,-\delta+1,\ldots,-1,0,1,\ldots,\delta-1,\delta\}$.
Therefore, dim $\mathcal{C}=n-2\delta-1$, and the minimum distance of $\mathcal{C}$ is at least $2\delta+2$ by
Lemma \ref{BCHBound_cyclic}. Then   $\mathcal{C}$ has parameters    $[n,n-2\delta-1,\geq2\delta+2]_{q^2}$.
Combining Corollary \ref{corollary_EAQECC}, Lemma \ref{Rank2} and the EA-Singleton
bound, we can obtain  an  EAQMDS code with parameters
$[[n,n-4\delta-1,2\delta+2;1]]_q$, where
 $2\leq 2\delta+2 \leq2\delta_{\max}+2=2\lfloor\frac{n}{q+1}\rfloor$.
In order to get EAQMDS codes with odd minimum distance, we take the
defining set of $\mathcal{C}$ as $Z=\cup_{i=-\delta'}^{\delta'-1}C_{i}$, where $1\leq\delta'\leq\delta_{\max}=\lfloor\frac{n}{q+1}\rfloor-1$. Then we can obtain  an   EAQMDS code with parameters
$[[n,n-4\delta'+1,2\delta'+1;1]]_q$, where
 $3\leq 2\delta'+1 \leq2\delta_{\max}+1=2\lfloor\frac{n}{q+1}\rfloor-1$.\qed
\end{proof}

\begin{corollary}
\label{EAQMDS22}
There exists an EAQMDS code  with parameters
\[
[[q^2-1,q^2-2d+2,d;1]]_q,
\]
where  $q$ is a prime power, $2\leq d\leq 2q-2$ is an integer.
\end{corollary}
\begin{example}
Let $q=5$, then $n=q^2-1=24$. Applying Corollary \ref{EAQMDS22}, we get four EAQMDS codes with minimum
distance greater than $q-1=4$ whose parameters are $[[24,17,5;1]]_5$, $[[24,15,6;1]]_5$, $[[24,13,7;1]]_5$, $[[24,11,8;1]]_5$.
\end{example}
\subsection{Length $n=q^2$}
\noindent

 Let   $\mathcal{RS}(n-1,r)$ denote   a RS code of length $n-1$ over  $\mathbb{F}_{q^2}$  with parameters $[n-1,n-r,r]$. We extend   $\mathcal{RS}(n-1,r)$ by adding an overall  parity check, and   denote the extended code   by  $\mathcal{\widehat{RS}}(n-1,r)$. Then $\mathcal{\widehat{RS}}(n-1,r)$  has parameters $[n,n-r,r+1]$. Let $\alpha$ be a primitive element of  $\mathbb{F}_{q^2}$ and let $(\alpha_1,\alpha_2,\ldots,\alpha_n)=(0,1,\ldots,\alpha^{n-2})$. Then the parity check matrix of $\mathcal{\widehat{RS}}(n-1,r)$ is given by
\begin{equation}
\label{parity check matrix of GRS codes}
H_{\mathcal{\widehat{RS}}(n-1,r)} =
\left(
\begin{array}{cccc}
1&1&\cdots&1\\
\alpha_1&\alpha_2&\cdots&\alpha_n\\
\vdots&\vdots&\vdots&\vdots\\
\alpha_1^{r-1}&\alpha_2^{r-1}&\cdots&\alpha_n^{r-1}
\end{array}
\right).
\end{equation}

\begin{lemma}
\label{rankofn_2}
If $q\leq r\leq2q-2$, then
the rank of $H_{\mathcal{\widehat{RS}}(n-1,r)}H_{\mathcal{\widehat{RS}}(n-1,r)}^\dag$
is equal to 1.
\end{lemma}
\begin{proof}
It is easy to find that  $1\leq r\leq q-1\Leftrightarrow$
$\mathcal{\widehat{RS}}(n-1,r)^{\bot_h}\subseteq
\mathcal{\widehat{RS}}(n-1,r)$ by \cite[Lemma 8]{grassl2004optimal}.
If $q\leq r\leq 2q-2$, then we have
\begin{eqnarray}
\nonumber
&&H_{\mathcal{\widehat{RS}}(n-1,r)}H_{\mathcal{\widehat{RS}}(n-1,r)}^\dag\\
&=&
\left(
\begin{array}{cccc}
1&1&\cdots&1\\
\alpha_1&\alpha_2&\cdots&\alpha_n\\
\vdots&\vdots&\vdots&\vdots\\
\alpha_1^{r-1}&\alpha_2^{r-1}&\cdots&\alpha_n^{r-1}
\end{array}
\right)\cdot
\left(
\begin{array}{cccc}
1&1&\cdots&1\\
\alpha_1^q&\alpha_2^q&\cdots&\alpha_n^q\\
\vdots&\vdots&\vdots&\vdots\\
\alpha_1^{q(r-1)}&\alpha_2^{q(r-1)}&\cdots&\alpha_n^{q(r-1)}
\end{array}
\right)^T\\
\label{1matrix}
&=&\left(
\begin{array}{cccccc}

0&0&\cdots&0&\cdots&0\\
\vdots&\vdots&\vdots&\vdots&\vdots&\vdots\\
0&0&\cdots& -1&\cdots&0\\
\vdots&\vdots&\vdots&\vdots&\vdots&\vdots\\
0&0&\cdots&0&\cdots&0
\end{array}
\right),
\end{eqnarray}
where  the ``-1"   in the $q$th row and $q$th column of matrix (\ref{1matrix})   is given by
\[
\alpha_1^{q^2-1}+\alpha_2^{q^2-1}+\ldots+\alpha_n^{q^2-1}=0+1+\ldots+1=-1.
\]
The zero elements of matrix (\ref{1matrix})  are given by
 \begin{eqnarray}
\nonumber
1+1+\ldots+1&=&0,\\
\nonumber
 \alpha_1^{r_1 }+\alpha_2^{r_1 }+\ldots+\alpha_n^{r_1 }&=&0,\\
 \nonumber
 \alpha_1^{ qr_2}+\alpha_2^{ qr_2}+\ldots+\alpha_n^{ qr_2}&=&0,\\
 \nonumber
 \alpha_1^{r_1+qr_2}+\alpha_2^{r_1+qr_2}+\ldots+\alpha_n^{r_1+qr_2}&=&0,
\end{eqnarray}
where $1\leq r_1, r_2 \leq r-1$, and then  $r_1$ and $r_2$ are not equal to $q-1$ simultaneously.
Therefore,
the rank of $H_{\mathcal{\widehat{RS}}(n-1,r)}H_{\mathcal{\widehat{RS}}(n-1,r)}^\dag$
is equal to 1. \qed
\end{proof}

Combining Corollary \ref{corollary_EAQECC} and   Lemma \ref{rankofn_2},   we can obtain the following EAQMDS code with length $q^2$.
\begin{theorem}
\label{n=q^2}
There exists an  EAQMDS code with parameters $[[q^2,q^2-2d+3,d;1]]_q$, where $q$ is a prime power, $q+1\leq d\leq2q-1$ is an integer.
\end{theorem}

\begin{example}
Let $q=5$, then $n=q^2=25$. Applying Thoerem \ref{n=q^2}, we get four EAQMDS codes with minimum
distance greater than $q=5$ whose parameters are $[[25,16,6;1]]_5$, $[[25,14,7;1]]_5$, $[[25,12,8;1]]_5$, $[[25,10,9;1]]_5$.
\end{example}

\subsection{EAQMDS codes that consume more than one maximally entangled states }
\noindent

In \cite{ kai2013new,kai2013constacyclic,wang2015new,chen2015application},    many QMDS codes have been constructed based on negacyclic  codes and constacyclic codes. If we  introduce a certain amount of extra pre-shared entanglement  in
some special cases, we can get   EAQMDS codes with larger minimum distance.

Let $q$ be an odd prime power and $n=\frac{q^2-1}{2}$. For $1\leq j\leq n$, it is easy to see that the $q^2$-ary cyclotomic coset
containing $2j-1$ modulo $2n$ has only one element $2j-1$, i.e., $C_{2j-1}=\{2j-1\}$.

\begin{lemma}
\label{lemma3.3.1}
Let $q$ be an odd prime power and $n=\frac{q^2-1}{2}$. Let $\mathcal{C}$ be a $q^2$-ary negacyclic code of length $n$ with defining
set $Z=\cup_{j=-\delta_1}^{\delta_2}C_{2j-1}$,
where  $1\leq\delta_1\leq \frac{q-1}{2}-1$ and $\frac{q+1}{2}\leq\delta_2\leq q-1$, and let  $H$ be
the parity check matrix of $\mathcal{C}$ over $\mathbb{F}_{q^2}$, then rank($HH^\dag)=2$.
\end{lemma}
\begin{proof}
We divide the defining set $Z$ of $\mathcal{C}$ into three mutually disjoint subsets, i.e.,  $Z=Z_1\cup C_{-1}\cup Z_2$, where $Z_1=\cup_{j=1}^{\delta_1}C_{-2j-1}$ and $Z_2=\cup_{j=1}^{\delta_2}C_{2j-1}$.
Let $\mathcal{C}_1$ and $\mathcal{C}_2$ be  two $q^2$-ary negacyclic codes of length $n$ with defining sets $Z_1$ and $Z_2$, respectively. We know that   $\mathcal{C}_1^{\bot_h}\subseteq \mathcal{C}_1$ and $\mathcal{C}_2^{\bot_h}\subseteq \mathcal{C}_2$ by   \cite[Lemma 3.1]{kai2013constacyclic}.
We show that $\mathcal{C}_1^{\bot_h}\subseteq \mathcal{C}_2$. Seeking a contradiction, we
assume that $Z_1\cap Z_2^{-q}\neq \emptyset$ by Lemma \ref{Hermitian_dual_coset}. Then there exist $k$ and $l$, where $1\leq k\leq  \frac{q-1}{2}-1$ and $1\leq l\leq q-1$, such that $-2k-1=-q(2l-1) \pmod{2n}$, which
means that $q(2l-1)-(2k+1)=0 \pmod{2n}$. It follows that $q(2l-1)-(2k+1)=q^2-1$ since $2\leq q(2l-1)-(2k+1)\leq 2q^2-3q-3$. However, there is $0\leq 2k=q(2l-q-1)\leq q^2-3q$. Then we have $k=0$ or $2k\geq 2q$, which are both contradictions.
Therefore, we have $\mathcal{C}_1^{\bot_h}\subseteq \mathcal{C}_2$. Let the parity check matrices of $\mathcal{C}_1$ and $\mathcal{C}_2$ over $\mathbb{F}_{q^2}$ be $H_1$ and $H_2$, respectively,
then we have $H_1H_1^\dag=0$, $H_2H_2^\dag=0$ and $H_1H_2^\dag=0$. It is easy to see that $C_{-1}\cap C_{-1}^{-q}=\emptyset$,
$C_{-1}\cap Z_1^{-q}=Z_1\cap C_{-1}^{-q}=\emptyset$,
$C_{-1}\cap Z_2^{-q}=\{-1\}$ and $Z_2\cap C_{-1}^{-q}= \{q\}$, hence, $h_{-1}h_{-1}^{\dagger}=0$, $h_{-1}H_1^{\dagger}=0$,
$H_1h_{-1}^{\dagger}=0$, $h_{-1}H_2^{\dagger}$ is a nonzero row vector and $H_2h_{-1}^{\dagger}$ is a nonzero column vector.
Then the parity check matrix of $\mathcal{C}$ over $\mathbb{F}_{q^2}$ is given by $H=\left(\begin{array}{c}H_1\\h_{-1}\\H_2\end{array}\right)$, where $h_{-1}=(1,\eta^{-1},\ldots,\eta^{-(n-1)})$. Then we have
 \begin{eqnarray*}
HH^\dagger &=&
\left(\begin{array}{ccc}
H_1H_1^{\dagger}&H_1h_{-1}^{\dagger}&H_1H_2^{\dagger}\\
h_{-1}H_1^{\dagger}&h_{-1}h_{-1}^{\dagger}&h_{-1}H_2^{\dagger}\\
H_2H_1^{\dagger}&H_2h_{-1}^{\dagger}&H_2H_2^{\dagger}\\
\end{array}\right)=\left(\begin{array}{ccc}
0&0&0\\
0&0&h_{-1}H_2^{\dagger}\\
0&H_2h_{-1}^{\dagger}&0\\
\end{array}\right).
\end{eqnarray*}
It follows that the rank of $HH^\dagger$ is equal to $2$. \qed
\end{proof}

\begin{theorem}
\label{theorem777}
There exists an EAQMDS code with parameters $[[\frac{q^2-1}{2},\frac{q^2-1}{2}-2d+4,d;2]]_q$, where $q$ is an odd prime power, $\frac{q+1}{2}+2\leq d\leq \frac{3}{2}q-\frac{1}{2}$.
\end{theorem}
\begin{proof}
Consider the negacyclic code $\mathcal{C}$ over $\mathbb{F}_{q^2}$ of
length $\frac{q^2-1}{2}$ with defining set $Z=\cup_{j=-\delta_1}^{\delta_2}C_{2j-1}$, where $0\leq\delta_1\leq \frac{q-1}{2}-1$ and $\frac{q+1}{2}\leq\delta_2\leq q-1$. Then
the defining set $Z$ which consists of $\delta_1+\delta_2+1$
consecutive odd integers is given by $\{-2\delta_1-1,-2\delta_1+1,\ldots,-1,1,\ldots,2\delta_2-3,2\delta_2-1\}$.
Therefore, we have dim $\mathcal{C}=\frac{q^2-1}{2}-\delta_1-\delta_2-1$. From the BCH bound for negacyclic codes in Lemma \ref{BCHBound_constacyclic}, the minimum distance
of $\mathcal{C}$ is at least $\delta_1+\delta_2+2$. Then $\mathcal{C}$ has parameters  $[\frac{q^2-1}{2},\frac{q^2-1}{2}-\delta_1-\delta_2-1,\geq\delta_1+\delta_2+2]_{q^2}$.
Combining Corollary \ref{corollary_EAQECC},  Lemma \ref{lemma3.3.1} and the EA-Singleton
bound, we can obtain  an  EAQMDS code with parameters
$[[\frac{q^2-1}{2},\frac{q^2-1}{2}-2\delta_1-2\delta_2-2,\delta_1+\delta_2+2;2]]_q$. Let $d=\delta_1+\delta_2+2$, we have $\frac{q+1}{2}+2\leq d\leq \frac{3}{2}q-\frac{1}{2}$.
\qed
\end{proof}

\begin{example}
Let $q=5$, then $n=\frac{q^2-1}{2}=12$. Applying Theorem \ref{theorem777}, we get three EAQMDS codes with parameters
$[[12,6,5;2]]_5$, $[[12,4,6;2]]_5$, $[[12,2,7;2]]_5$.
\end{example}

Let $t\geq3$ be an odd integer and let $q$ be an odd prime power with $t|(q+1)$. Denote   $n=\frac{q^2-1}{t}$. Let $\lambda\in\mathbb{F}_{q^2}$ be a primitive $t$-th root of unity. It is easy to see that every    $q^2$-cyclotomic coset modulo $tn$ contains only one element.  In \cite{wang2015new,chen2015application}, $q$-ary QMDS codes of length $n=\frac{q^2-1}{t}$ have been constructed from Hermitian dual-containing $\lambda$-constacyclic MDS codes. Based on the $\lambda$-constacyclic MDS codes, and if we introduce  a certain amount of extra pre-shared entanglement, we can get     EAQMDS codes with larger minimum distance  compared with      QMDS codes in \cite{wang2015new,chen2015application} of   length $n=\frac{q^2-1}{t}$.
 Let $\mathcal{C}$ be a $\lambda$-constacyclic code of length $n$ over $\mathbb{F}_{q^2}$ with defining set
\begin{equation}
Z=\cup_{i=-\delta_1}^{\delta_2}C_{1+t(\frac{(t-1)(q-1)-2}{2t}+i)},
\end{equation}
where $C_{1+t(\frac{(t-1)(q-1)-2}{2t}+i)}=\{1+t(\frac{(t-1)(q-1)-2}{2t}+i)\}$ for $-\delta_1\leq i\leq\delta_2$, $\frac{(t-1)(q+1)}{2t} \leq\delta_1\leq\frac{(t+1)(q+1)}{2t}-2 $ and $\frac{(t-1)(q+1)}{2t} \leq\delta_2\leq\frac{(t+1)(q+1)}{2t}-2$.
\begin{lemma}
\label{lemma3.3.2}
Let $t\geq 3$ be an odd integer and let $q$ be an odd prime power with $t|(q+1)$. Denote  $n=\frac{q^2-1}{t}$. Let $\mathcal{C}$ be a $q^2$-ary $\lambda$-constacyclic code of length $n$ with defining
set $Z=\cup_{i=-\delta_1}^{\delta_2}C_{1+t(\frac{(t-1)(q-1)-2}{2t}+i)}$, where
$\frac{(t-1)(q+1)}{2t} \leq\delta_1\leq\frac{(t+1)(q+1)}{2t}-2 $ and $\frac{(t-1)(q+1)}{2t} \leq\delta_2\leq\frac{(t+1)(q+1)}{2t}-2$,
and let  $H$ be
the parity check matrix of $\mathcal{C}$ over $\mathbb{F}_{q^2}$,  then rank($HH^\dag)=t$.
\end{lemma}
\begin{proof}
Denote   $s=(t-1)/2$. We can divide the defining set $Z$ of $\mathcal{C}$ into three mutually disjoint subsets, i.e.,  $Z=Z_1\cup C_{s(q-1)}\cup Z_2$, where $Z_1=\cup_{j=1}^{\delta_1}C_{1+t(\frac{(t-1)(q-1)-2}{2t}-j)}$ and $Z_2=\cup_{k=1}^{\delta_2}C_{1+t(\frac{(t-1)(q-1)-2}{2t}+k)}$.
Let $\mathcal{C}_1$ and $\mathcal{C}_2$ be two $q^2$-ary $\lambda$-constacyclic codes of length $n$ with defining sets $Z_1$ and $Z_2$, respectively.  Let the parity check matrices of $\mathcal{C}_1$ and $\mathcal{C}_2$
over $\mathbb{F}_{q^2}$ be $H_1$ and $H_2$, respectively.
Then the parity check matrix of $\mathcal{C}$ is given by $H=\left(\begin{array}{c}H_1\\h_{q-1}\\H_2\end{array}\right)$, where $h_{q-1}=(1,\eta^{q-1},\ldots,\eta^{(n-1)(q-1)})$.
From \cite[Lemma 3.6]{wang2015new} and \cite[Lemma 4.1]{chen2015application},
there are $\mathcal{C}_1^{\bot_h}\subseteq \mathcal{C}_1$ and $\mathcal{C}_2^{\bot_h}\subseteq \mathcal{C}_2$, hence
$H_1H_1^\dagger=0$ and $H_2H_2^\dagger=0$. It is easy to see that $C_{s(q-1)}\cap C_{s(q-1)}^{-q}=\{s(q-1)\}$, $Z_1\cap C_{s(q-1)}^{-q}=C_{s(q-1)}\cap Z_1^{-q}=\emptyset$ and $Z_2\cap C_{s(q-1)}^{-q}=C_{s(q-1)}\cap Z_2^{-q}=\emptyset$, then there are $h_{-1}h_{-1}^{\dagger}= 1+1+\cdots+1 \neq 0$, $H_1h_{-1}^{\dagger}=0$, $h_{-1}H_1^{\dagger}=0$
and $h_{-1}H_2^{\dagger}=0$, $H_2h_{-1}^{\dagger}=0$. Then we have
\begin{eqnarray}
HH^\dagger &=&
\left(\begin{array}{ccc}
H_1H_1^{\dagger}&H_1h_{-1}^{\dagger}&H_1H_2^{\dagger}\\
h_{-1}H_1^{\dagger}&h_{-1}h_{-1}^{\dagger}&h_{-1}H_2^{\dagger}\\
H_2H_1^{\dagger}&H_2h_{-1}^{\dagger}&H_2H_2^{\dagger}\\
\end{array}\right)=\left(\begin{array}{ccc}
0&0&H_1H_2^{\dagger}\\
0&h_{-1}h_{-1}^{\dagger}&0\\
H_2H_1^{\dagger}&0&0\\
\end{array}\right).
\end{eqnarray}

It follows that $rank(HH^\dagger)=2rank(H_1H_2^{\dagger})+1$. Next, we have to compute the rank of $H_1H_2^{\dagger}$. We determine the intersection  of $Z_1$ and $Z_2^{-q}$. We assume that there exist $j$ and $k$, where $ 1\leq j\leq\frac{(t+1)(q+1)}{2t}-2 $ and $ 1\leq k\leq\frac{(t+1)(q+1)}{2t}-2$,
such that $1+t(\frac{(t-1)(q-1)-2}{2t}-j)=-q(1+t(\frac{(t-1)(q-1)-2}{2t}+k)) \pmod{q^2-1}$, which means that $tqk-tj=0 \pmod{q^2-1}$. Since $\frac{t-1}{2}q +\frac{3t-1}{2}\leq tqk-tj\leq\frac{t+1}{2}q^2-\frac{3t-1}{2}q-t$,
it follows that $tqk-tj\in\{q^2-1,\ldots, \frac{t-1}{2}(q^2-1)\}$. Denote   $tqk-tj=x_t(q^2-1)$, where $x_t\in\{1,\ldots,s\}$, then we have $\frac{x_t(q^2-1)+t}{tq}\leq k\leq\frac{2x_t(q^2-1)+(t+1)q-3t+1}{2tq}$.
Note that $\frac{x_t(q+1)}{t}-1<\frac{x_t(q^2-1)+t}{tq}\leq k\leq\frac{2x_t(q^2-1)+(t+1)q-3t+1}{2tq}<\frac{x_t(q+1)}{t}+1$. It follows that  $k=\frac{x_t(q+1)}{t}$ and $j=\frac{x_t(q+1)}{t}$ for $x_t\in\{1,\ldots,s\}$. Therefore, we have $Z_1\cap Z_2^{-q}=\{
\frac{(t-2x_t-1)q-2x_t-t-1}{2t}|x_t=1,\ldots,s\}$   and $|Z_1\cap Z_2^{-q}|=s$. We can redivide $Z_1$ and $Z_2$ into
mutually disjoint subsets, respectively, then the rank of $H_1H_2^\dagger$ is equal to $s$. Therefore, $rank(HH^\dagger)=2\cdot s+1=t$. \qed
\end{proof}

\begin{theorem}
\label{theorem88}
Let $t\geq3$ be an odd integer and let $q$ be an odd prime power with $t|(q+1)$. Then, there exists an EAQMDS code  with parameters $[[\frac{q^2-1}{t},\frac{q^2-1}{t}-2d+t+2,d;t]]_q$, where
$\frac{(t-1)(q+1)}{t}+2 \leq d\leq  \frac{(t+1)(q+1)}{ t}-2$.
\end{theorem}
\begin{proof}
Let  $\mathcal{C}$ be a $\lambda$-constacyclic code  over $\mathbb{F}_{q^2}$ of
length $\frac{q^2-1}{t}$ with defining set $Z=\cup_{i=-\delta_1}^{\delta_2}C_{1+t(\frac{(t-1)(q-1)-2}{2t}+i)}$, where
$\frac{(t-1)(q+1)}{2t} \leq\delta_1\leq\frac{(t+1)(q+1)}{2t}-2 $ and $\frac{(t-1)(q+1)}{2t} \leq\delta_2\leq\frac{(t+1)(q+1)}{2t}-2 $.
Note that  dim $\mathcal{C}=\frac{q^2-1}{t}-\delta_1-\delta_2-1$, and the minimum distance
of $\mathcal{C}$ is at least $\delta_1+\delta_2+2$ by the BCH bound for constacyclic codes in Lemma  \ref{BCHBound_constacyclic}. Then $\mathcal{C}$  has parameters $[\frac{q^2-1}{t},\frac{q^2-1}{t}-\delta_1-\delta_2-1,\geq d]_{q^2}$,
where $d=\delta_1+\delta_2+2$. Combining Corollary \ref{corollary_EAQECC}, Lemma \ref{lemma3.3.2} and the EA-Singleton
bound,
we can obtain  an  EAQMDS code
with parameters
$[[\frac{q^2-1}{t},\frac{q^2-1}{t}-2\delta_1-2\delta_2+1,d;t]]_q$, where
$\frac{(t-1)(q+1)}{t}+2 \leq d\leq  \frac{(t+1)(q+1)}{ t}-2$.
\qed
\end{proof}
\begin{example}
Let $t=3$ and $q=11$, then $n=\frac{q^2-1}{3}=40$. We get five EAQMDS codes with parameters
$[[40,25,10;3]]_{11}$, $[[40,23,11;3]]_{11}$, $[[40,21,12;3]]_{11}$, $[[40,19,13;3]]_{11}$, $[[40,17,14;3]]_{11}$.
\end{example}
\begin{example}
Let $t=5$ and $q=19$, then $n=\frac{q^2-1}{5}=72$. We get five EAQMDS codes with parameters
$[[72,43,18;5]]_{19}$, $[[72,41,19;5]]_{19}$, $[[72,39,20;5]]_{19}$, $[[72,37,21;5]]_{19}$, $[[72,35,22;5]]_{19}$.
\end{example}
\begin{example}
Let $t=7$, $q=27$, then $n=\frac{q^2-1}{7}=104$. We get five EAQMDS codes with parameters
$[[104,61,26;7]]_{27}$, $[[104,59,27;7]]_{27}$, $[[104,57,28;7]]_{27}$, $[[104,55,29;7]]_{27}$, $[[104,53,30;$ $7]]_{27}$.
\end{example}

\newcommand{\tabincell}[2]{\begin{tabular}{@{}#1@{}}#2\end{tabular}}
\begin{table}[tp]
\renewcommand{\arraystretch}{1.2}
\centering
\small
\tcaption{Comparison between  EAQMDS codes and standard QMDS codes\label{Comparison_QMDS}}
\begin{tabular}{|c|c|c|c|}
\hline
Length & $q$-ary EAQMDS codes& $q$-ary QMDS codes &Reference \\
\hline
$q^2+1$&\tabincell{c}{$[[q^2+1,q^2-2d+4,d;1]]$,
\\$2\leq d\leq2q$, $d$ even}&\tabincell{c}{$[[q^2+1,q^2-2d+3,d]]$,\\$2\leq d\leq q+1$}&\tabincell{c}{\cite{jin2013construction}, \cite{jin2010application},\\\cite{la2011new}, \cite{kai2013new}}\\
\hline
$q^2$&
\tabincell{c}{$[[q^2,q^2-2d+3,d;1]]$,\\$q+1\leq d\leq2q-1$}&\tabincell{c}{$[[q^2,q^2-2d+2,d]]$,\\$2\leq d\leq q$}&\cite{li2008quantum}, \cite{grassl2004optimal}\\
\hline
$q^2-1$& \tabincell{c}{$[[q^2-1,q^2-2d+2,d;1]]$,\\$2\leq d\leq2q-2$}&\tabincell{c}{$[[q^2-1,q^2-2d+1,d]]$,\\$2\leq d\leq q-1$}&\cite{li2008quantum}, \cite{grassl2004optimal}\\
\hline
\tabincell{c}{$\frac{q^2-1}{2}$, $q$ odd}&\tabincell{c}{$[[\frac{q^2-1}{2},
\frac{q^2-1}{2}-2d+4,d;2]]$,\\$(q+1)/2+2\leq d\leq\frac{3}{2}q-\frac{1}{2}$}&\tabincell{c}{$[[\frac{q^2-1}{2},\frac{q^2-1}{2}-2d+2,d]]$,\\
$2\leq d\leq q$}&\cite{kai2013constacyclic}, \cite{wang2015new}\\
\hline
\tabincell{c}{$\frac{q^2-1}{t}$, $q$ odd,\\$t|(q+1)$,\\$t\geq3$ odd}&\tabincell{c}{$[[\frac{q^2-1}{t},
\frac{q^2-1}{t}-2d+t+2,d;t]]$,\\$\frac{(t-1)(q+1)}{t}+2 \leq d\leq \frac{(t+1)(q+1)}{t}-2$}&\tabincell{c}{$[[\frac{q^2-1}{t},\frac{q^2-1}{t}-2d+2,d]]$,\\
$2\leq d\leq \frac{(t+1)(q+1)}{2t}-1$}&\cite{wang2015new}, \cite{chen2015application}\\
\hline
\end{tabular}
\end{table}
\section{Conclusion}

\noindent

We have constructed several classes of  entanglement-assisted quantum MDS (EAQMDS) codes
based on classical MDS codes for some certain code lengths. We list a comparison   in Table \ref{Comparison_QMDS} between EAQMDS codes constructed in this paper and the standard QMDS codes. Compared with the known QMDS codes of the same length,
these EAQMDS codes have much larger minimum distance upper limit by exploiting one or more pre-shared maximally entangled states.
In the future work, we look forward to getting more $q$-ary EAQMDS codes  with   minimum distance greater than $q+1$.

\nonumsection{Acknowledgements}
\noindent

The authors are grateful to the  Editor and the anonymous
referee for their constructive comments and valuable
suggestions. The first author J. Fan thanks  the financial support from China Scholarship Council  (CSC, No. 201406090079). J. Fan thanks  Dr. Bocong Chen for the helpful communication. This work was supported   by the National Natural Science Foundation of China (Grant No. 61170321), the Specialized Research Fund for the Doctoral Program of Higher Education (Grant No. 20110092110024), the Natural Science Foundation of Jiangsu Province (Grant No. BK20140823), China Postdoctoral Science Foundation (Grant No. 2013M531353) and the Scientific  Research Innovation Plan for College Graduates of
Jiangsu Province (Grant No. CXZZ13\_0105).  This work
was partially carried out when the first author was visiting the School of Electrical and Information Engineering, University of Sydney.
He thanks the school for its hospitality.

\nonumsection{References}
\noindent

\bibliographystyle{IEEEtran}
\bibliography{IEEEabrv,myQIC140925}

\begin{thebibliography}{10}
\providecommand{\url}[1]{#1}
\csname url@samestyle\endcsname
\providecommand{\newblock}{\relax}
\providecommand{\bibinfo}[2]{#2}
\providecommand{\BIBentrySTDinterwordspacing}{\spaceskip=0pt\relax}
\providecommand{\BIBentryALTinterwordstretchfactor}{4}
\providecommand{\BIBentryALTinterwordspacing}{\spaceskip=\fontdimen2\font plus
\BIBentryALTinterwordstretchfactor\fontdimen3\font minus
  \fontdimen4\font\relax}
\providecommand{\BIBforeignlanguage}[2]{{%
\expandafter\ifx\csname l@#1\endcsname\relax
\typeout{** WARNING: IEEEtran.bst: No hyphenation pattern has been}%
\typeout{** loaded for the language `#1'. Using the pattern for}%
\typeout{** the default language instead.}%
\else
\language=\csname l@#1\endcsname
\fi
#2}}
\providecommand{\BIBdecl}{\relax}
\BIBdecl

\bibitem{calderbank1998quantum}
A.~R. Calderbank, E.~M. Rains, P.~Shor, and N.~J. Sloane, ``Quantum error
  correction via codes over {GF}(4),'' \emph{IEEE Trans. Inform. Theory},
  vol.~44, no.~4, pp. 1369--1387, 1998.

\bibitem{gottesman1997stabilizer}
D.~Gottesman, ``Stabilizer codes and quantum error correction,'' Ph.D.
  dissertation, California Institute of Technology, 1997.

\bibitem{ketkar2006nonbinary}
A.~Ketkar, A.~Klappenecker, S.~Kumar, and P.~K. Sarvepalli, ``Nonbinary
  stabilizer codes over finite fields,'' \emph{IEEE Trans. Inform. Theory},
  vol.~52, no.~11, pp. 4892--4914, 2006.

\bibitem{brun2006correcting}
T.~Brun, I.~Devetak, and M.-H. Hsieh, ``Correcting quantum errors with
  entanglement,'' \emph{Science}, vol. 314, no. 5798, pp. 436--439, 2006.

\bibitem{hsieh2009entanglement}
M.-H. Hsieh, T.~A. Brun, and I.~Devetak, ``Entanglement-assisted quantum
  quasicyclic low-density parity-check codes,'' \emph{Phys. Rev. A}, vol.~79,
  no.~3, p. 032340, 2009.

\bibitem{hsieh2011high}
M.-H. Hsieh, W.-T. Yen, and L.-Y. Hsu, ``High performance entanglement-assisted
  quantum {LDPC} codes need little entanglement,'' \emph{IEEE Trans. Inform.
  Theory}, vol.~57, no.~3, pp. 1761--1769, 2011.

\bibitem{fujiwara2010entanglement}
Y.~Fujiwara, D.~Clark, P.~Vandendriessche, M.~De~Boeck, and V.~D. Tonchev,
  ``Entanglement-assisted quantum low-density parity-check codes,'' \emph{Phys.
  Rev. A}, vol.~82, no.~4, p. 042338, 2010.

\bibitem{fujiwara2013characterization}
Y.~Fujiwara and V.~D. Tonchev, ``A characterization of entanglement-assisted
  quantum low-density parity-check codes,'' \emph{IEEE Trans. Inform. Theory},
  vol.~59, no.~6, pp. 3347--353, 2013.

\bibitem{wilde2014entanglement}
M.~M. Wilde, M.-H. Hsieh, and Z.~Babar, ``Entanglement-assisted quantum turbo
  codes,'' \emph{IEEE Trans. Inform. Theory}, vol.~60, no.~2, pp. 1203--1222,
  2014.

\bibitem{lu2014entanglement}
L.-D. L{\"u} and R.~Li, ``Entanglement-assisted quantum codes constructed from
  primitive quaternary {BCH} codes,'' \emph{Int. J. Quantum Inf.}, vol.~12,
  no.~03, p. 1450015, 2014.

\bibitem{lidar2013quantum}
D.~A. Lidar and T.~A. Brun, \emph{Quantum error correction}.\hskip 1em plus
  0.5em minus 0.4em\relax Cambridge: Cambridge University Press, 2013.

\bibitem{wilde2013quantum}
M.~M. Wilde, \emph{Quantum Information Theory}.\hskip 1em plus 0.5em minus
  0.4em\relax Cambridge: Cambridge University Press, 2013.

\bibitem{macwilliams1977theory}
F.~J. MacWilliams and N.~J.~A. Sloane, \emph{The Theory of Error-Correcting
  Codes}.\hskip 1em plus 0.5em minus 0.4em\relax Amsterdam: The Netherlands:
  North-Holland, 1981.

\bibitem{grassl1999quantum}
M.~Grassl, W.~Geiselmann, and T.~Beth, ``Quantum {Reed-Solomon} codes,'' in
  \emph{Applied Algebra, Algebraic Algorithms and Error-correcting
  Codes}.\hskip 1em plus 0.5em minus 0.4em\relax Springer, 1999, pp. 231--244.

\bibitem{li2008quantum}
Z.~Li, L.-J. Xing, and X.-M. Wang, ``Quantum generalized {Reed-Solomon} codes:
  Unified framework for quantum maximum-distance-separable codes,'' \emph{Phys.
  Rev. A}, vol.~77, no.~1, p. 012308, 2008.

\bibitem{grassl2015quantum}
M.~Grassl and M.~Roetteler, ``Quantum {MDS} codes over small fields,'' in
  \emph{Proc. {IEEE} Int. Symp. Inf. Theory}, Hong Kong, June 2015, pp.
  1104--1108.

\bibitem{jin2013construction}
L.~Jin and C.~Xing, ``A construction of new quantum {MDS} codes,'' \emph{IEEE
  Trans. Inform. Theory}, vol.~60, no.~5, pp. 2921--2925, 2014.

\bibitem{rotteler2004quantum}
M.~R{\"o}tteler, M.~Grassl, and T.~Beth, ``On quantum {MDS} codes,'' in
  \emph{Proc. {IEEE} Int. Symp. Inf. Theory}, Chicago, IL, USA, June 2004, pp.
  356--356.

\bibitem{grassl2004optimal}
M.~Grassl, T.~Beth, and M.~Roetteler, ``On optimal quantum codes,'' \emph{Int.
  J. Quantum Inf.}, vol.~2, no.~01, pp. 55--64, 2004.

\bibitem{li2010construction}
R.~Li and Z.~Xu, ``Construction of $[[n, n-4, 3]]_q$ quantum codes for odd
  prime power $q$,'' \emph{Phys. Rev. A}, vol.~82, no.~5, p. 052316, 2010.

\bibitem{jin2010application}
L.~Jin, S.~Ling, J.~Luo, and C.~Xing, ``Application of classical hermitian
  self-orthogonal {MDS} codes to quantum {MDS} codes,'' \emph{IEEE Trans.
  Inform. Theory}, vol.~56, no.~9, pp. 4735--4740, 2010.

\bibitem{la2011new}
G.~G. La~Guardia, ``New quantum {MDS} codes,'' \emph{IEEE Trans. Inform.
  Theory}, vol.~57, no.~8, pp. 5551--5554, 2011.

\bibitem{kai2013new}
X.~Kai and S.~Zhu, ``New quantum {MDS} codes from negacyclic codes,''
  \emph{IEEE Trans. Inform. Theory}, vol.~59, no.~2, pp. 1193--1197, 2013.

\bibitem{kai2013constacyclic}
X.~Kai, S.~Zhu, and P.~Li, ``Constacyclic codes and some new quantum {MDS}
  codes,'' \emph{IEEE Trans. Inform. Theory}, vol.~60, no.~4, pp. 2080--2086,
  2014.

\bibitem{zhang2014new}
G.~Zhang and B.~Chen, ``New quantum {MDS} codes,'' \emph{Int. J. Quantum Inf.},
  vol.~12, no.~04, 2014.

\bibitem{wang2015new}
L.~Wang and S.~Zhu, ``New quantum {MDS} codes derived from constacyclic
  codes,'' \emph{Quantum Inf. Process.}, vol.~14, no.~3, pp. 881--889, 2015.

\bibitem{chen2015application}
B.~Chen, S.~Ling, and G.~Zhang, ``Application of constacyclic codes to quantum
  {MDS} codes,'' \emph{IEEE Trans. Inform. Theory}, vol.~61, no.~3, pp.
  1474--1484, 2015.

\bibitem{berlekamp1968algebraic}
E.~Berlekamp, \emph{Algebraic Coding Theory}.\hskip 1em plus 0.5em minus
  0.4em\relax New York, McGraw-Hill, 1968.

\bibitem{wilde2008optimal}
M.~M. Wilde and T.~A. Brun, ``Optimal entanglement formulas for
  entanglement-assisted quantum coding,'' \emph{Phys. Rev. A}, vol.~77, no.~6,
  p. 064302, 2008.

\bibitem{krishna1990pseudocyclic}
A.~Krishna and D.~V. Sarwate, ``Pseudocyclic maximum-distance-separable
  codes,'' \emph{IEEE Trans. Inform. Theory}, vol.~36, no.~4, pp. 880--884,
  1990.

\bibitem{nielsen2000quantum}
M.~A. Nielsen and I.~L. Chuang, \emph{Quantum Computation and Quantum
  Information}.\hskip 1em plus 0.5em minus 0.4em\relax Cambridge: Cambridge
  University Press, 2000.

\end{thebibliography}

\end{document}